\documentclass{llncs}
\usepackage{llncsdoc}
\usepackage{amsmath}
\usepackage{amsfonts}
\usepackage{amssymb}
\usepackage{graphicx}

\usepackage{authblk}
\usepackage{tkz-euclide}
\usepackage{tikz}
\usepackage{subfigure}
\usepackage{balance}  
\usepackage[linesnumbered,ruled]{algorithm2e}
\usepackage{multirow}
\usepackage{booktabs}
\usepackage{url}

\usepackage{footmisc}
\DefineFNsymbols{mySymbols}{{\ensuremath\dagger}{\ensuremath\ddagger}\S\P
   *{**}{\ensuremath{\dagger\dagger}}{\ensuremath{\ddagger\ddagger}}}
\setfnsymbol{mySymbols}

\begin{document}
\markboth{\LaTeXe{} Class for Lecture Notes in Computer
Science}{\LaTeXe{} Class for Lecture Notes in Computer Science}
\thispagestyle{empty}
\title{Ranking Object under Team Context}
\author{Xiaolu Lu\inst{1} \and  Dongxu Li\inst{1} \and Xiang Li\thanks{Xiang Li is the corresponding author.}\inst{1} \and Ling Feng\inst{2}}
\institute{School of Software, Nanjing University
\email{\{mf1232050,mf1332027,lx\}@software.nju.edu.cn}
 \and
Department of Computer Science and Technology, Tsinghua University
\email{fengling@tsinghua.edu.cn}}
\maketitle
\begin{abstract}
Context-aware database has drawn increasing attention from both industry and academia recently by taking users'
current situation and environment into consideration. However, most of the literature focus on individual context,
overlooking the team users. In this paper, we investigate how to integrate team context into database query process to
help the users' get top-ranked database tuples and make the team more competitive.  We review naive and propose an optimized query
algorithm to select the suitable records and show that they output the same results while the latter is more computational efficient.
Extensive empirical studies are conducted to evaluate the query approaches and demonstrate their effectiveness and efficiency.
\end{abstract}

\section{Introduction}
Millions of users take portable devices in the palm of their hands. It leads to the rapid development of context-aware database whose users have great expectations of getting suitable query results based on their ambient environment. At the same time, context-aware query has been widely explored to tackle with the many-answers problem to get rid of overwhelming information. Essentially, these applications keep context information to predict users' preferences. Researches in context-aware query have mainly focused on contexts
from sensors and user profiles rather than the users' organization-level context i.e the team context.
Recently, the problem of context-aware database
query has drawn increasing attention from both industry and academia.
To cope with the problem many approaches have been proposed
and can be divided into two categories: \emph{qualitative} and \emph{quantitative}.
Qualitative approaches model the user preference as partial
order and apply logic tools to reason the user's intention \cite{DBLP:conf/dexa/BunningenFA06}.
On the other hand, quantitative approaches compute the users' satisfaction by score function \cite{DBLP:conf/eurossc/LiFZ08}.
However, most of them are based on the individual context. In \cite{DBLP:conf/icdim/LiF12a}, group context is taken but the group cannot change during the query process.
In this paper, we propose the problem of
ranking under team context(RTC) which queries
the database from a team's perspective and aims
to helping the users have a more competitive context
by ranking and replace some team component with top-ranked tuples.
For example, NBA teams are preparing the roster and select the prospective player in hope of qualify for the play-offs(finals of NBA) in the next season. To this end, we need to consider the team context in a united group to query the player database for the best player,i.e., who does the team need to acquire in the coming season to make itself a serious candidate for play-offs and which player in the current team should be included in a trade. To the best of our knowledge, this work is the first to focus on team context query while traditional context-aware methods relies on individual context. When taking the whole organization background into consideration, querying becomes more practical and convenient for company customers. Moreover, team context-aware query make it easy to get different query results from
different layers of hierarchy which meets perfectly with the innate characteristics of many contexts.
The brute force method can be quite inefficient due to excessive I/O overhead. In an effort to handle the limitation, we introduce an I/O-efficient approach RTC* based on Nearest Neighbour. RTC* calculate the exact virtual component the user need to replace with,and map it to the database space. With nearest neighbour technique, we offer the ranking of query results. We prove that RTC* can produce the same results as the brute force method.

We summarize our key contributions as follows:
\begin{itemize}
  \item We define the RTC problem.
  \item We propose the solution to the RTC problem based on NN-indexing and prove its correctness.
  \item We evaluate our algorithms by experiments in terms of effectiveness and efficiency.
\end{itemize}
The rest of the paper is organized as follows:
Section \ref{sec:Related Works} describes related work with a comparison. Section \ref{sec:Prelimilaries} defines the RTC problem and section \ref{sec:Solutions to RTC Problem} proposes our method with a review of baseline method. Section \ref{sec:Experiment} presents the experiments with analysis of the results. Section \ref{sec:ConclusionLB} is conclusion and future work.

\section{Related Work}\label{sec:Related Works}
Object ranking under team context is a kind of context-aware query processing, aiming at helping systems provide query results after understanding the real intentions behind the queries.  To be more specific, our work handled the context that has a team property with a goal of being more competitive and approaching teams with higher rank.
Researches in field of context-aware query can be roughly divided into two categories: qualitative and quantitative.
In qualitative strategies: preference over database tuples are calculated by score functions. \cite{DBLP:conf/eurossc/LiFZ08,DBLP:conf/icde/StefanidisPV07}
In quantitative strategies: logical rules are hard coded to database system to infer the users' preferences. \cite{DBLP:conf/icde/BunningenFAF07}
But group or team context is overlooked for quite a long time.
Recently, researches on group preferences have been reported. In \cite{DBLP:conf/er/StefanidisSNK12},
Stefanids \emph{et al.} generalized their previous work on hierarchical context model to tackle the needs of a group.
In \cite{DBLP:conf/icdim/LiF12a}, Li and Feng propose several methods to meet most of the people's contexts. However, all these work consider group as union of individuals or most of the members. In our work, team context is take in its entirety. The object selection based on team context is to make the team closer to its rivals. Context tackled in this paper is formed by objects from the object space, which have not been exploited.

\textit{k}-NN algorithm was one of the most widely used approaches in many fields,
first proposed by \cite{altman1992introduction} and continuously improving and refining for
specific purpose, especially in spatial databases and sensor networks. Moreover, applying \textit{k}-NN
approach in high dimensional data has raised many attentions. \cite{DBLP:conf/adc/HuCS04}
proposed a new method for performing data processing using \textit{k}-NN in high dimensional data and
provided a lower bound of distance between feature vectors. \cite{DBLP:conf/pci/KouiroukidisE11,DBLP:journals/infsof/YuCWS07}
reviewed and put forward a method with hybrid index techniques for solving so called
"curse of dimensionality" problems. The state-of-art high dimensional indexing technique iDistance proposed by  \cite{DBLP:journals/tods/JagadishOTYZ05} to enhance efficiency of existing approaches. Recently,  \cite{DBLP:conf/cikm/ZhongLTZ13} propose G-tree index for finding the \textit{k} nearest objects to the given location. \cite{DBLP:conf/bncod/SchuhWBA13} has carefully reviewed the skills in partitioning the data space by iDistance.

\section{Prelimilaries}\label{sec:Prelimilaries}
\subsection{A Motivated Example}\label{sec:a motivated example}
Consider an example in NBA. A fact is that if the games winning of one team ranks top 10 in regular seasons, it would be guaranteed to enter into play-offs. What should a team ranked 11st$\scriptsize{\sim}$20th do for entering into play-offs?

Assume a team $C$ ranked 17th in NBA wants to enter into play-offs. From the team's view, if the team could approach or even supersede one of the top 10 teams, its chance for entering into play-offs will becomes greater. We refer the team to be surpassed by current one as the target team.
 
To achieve this goal, usually one player in $C$ will be exchanged with another bought in the transaction. Which pair of players should be selected for fulfilling this goal is a challenging question needed to be answered.

Similar scenarios will also occurred in other area, such as in teams of software developers, clusters of computers etc. Motivated by those ones, Problem solved in this paper can be interpreted as rank the objects and select the ones served as the substitution of a objects in the team.

\subsection{Problem Formulation}\label{sec:problem formulation}
Given an object space $\mathbb{O}$ with \textit{n} $d$-dimensional objects. Team context(TC) in our paper is defined as a context $C$ formed by \textit{m} objects $\{O_1,O_2,..,O_m,\small \forall  O_m \in C, \small O_m \in \mathbb{O}\}$ like how teams formed in NBA. Also, define a target team $T$ of $C$ for approaching. Clearly, the contributions of each object differs according to different team contexts, like performance varies of one player in different teams in NBA. Thus, while exchanging objects, a set of exchanging parameters $\Lambda\{\lambda_1,\lambda_2,...,\lambda_m\}$ is defined for measuring the contribution of $O \in \mathbb{O}$ under current TC.
Formally, our problem is:
\begin{problem}\label{the RTC Problem}
\underline{\emph{Ranking under Team Context}(RTC Problem)} Rank the objects in $\mathbb{O}$ and determine a swap-in object which is top-ranked corresponding to a swap-out object in $C$. After performing the exchanging procedure, $C$ can approach $T$ to its best effort.
\end{problem}




\section{Solutions to RTC Problem}\label{sec:Solutions to RTC Problem}
\subsection{Modelling the Problem}\label{sec:model problem}
\subsubsection{Contributions of Objects}\label{sec:contribution of objects}
Since team context $C$ is formed by objects in $\mathbb{O}$, $C$ can also be described by contributions of its components. There exist different ways calculating the contributions of components, which are based on how different contexts are organized. In this paper, we adopt the method which means the team's ability is the accumulation of all its components, since it is the most widely used way in real scenarios. Demonstrate in \eqref{eq:summation of attribute}:
\begin{equation}\label{eq:summation of attribute}
c_i = \sum_{j=1}^m{o_{ji}}
\end{equation}
\noindent where $c_{i}$ means value on \textit{i}th dimension of $C$ and $o_{ji}$ means the value of \textit{j}th object on its \textit{i}th dimension.
\subsubsection{Contributions of Attributes}\label{sec:contributions of attriutes}
Although final ranking of the one team depends on values of all attributes, not all of them weight equally. For identifying the importance of each attribute, we adopt the Kendall's tau($\tau$) coefficients, which is a rank coefficient measuring association between two measured associations \cite{kendall1938new}.

Through calculating the association between each dimension and final ranking of team context pairwisely, an coefficient will be obtained and will be regarded as weight parameter of corresponding dimension \textit{i}, denote as $w_i$.
\subsubsection{Truncated Distance} \label{sec:distance between tc} Usually, we measure the difference between contexts(or objects) by weighted Euclid distance. However, positive distance yields to represent the overall conditions of a TC, especially the case shown in Fig.\ref{fig:subfig:truncated distance case1}.

The 2-D case in Fig.\ref{fig:subfig:truncated distance case1} depicted that $C$ exceeds $T$ in dimension $x$ while yields $T$ in dimension $y$. If we only consider measuring the distance by weighted Euclid distance, $C$ might be far away from $T$ due to abstract advantage on dimension $x$, therefore, situation such as losing strength on $x$ dimension when approaching $T$ will occur, which contradicts the team's goal. In order to preserve advantage of $C$ while approaching $T$, using truncated distance as a measurement is adopted as shown in Fig.\ref{fig:subfig:truncated distance case3}. As illustrated, if Case 1 happens, we only consider the distance between $C'$ and $T$ rather than $C$ and $T$. Another case shown in Fig.\ref{fig:subfig:truncated distance case2} is relatively simple for tackling since $C$ lags $T$ in both dimension $x$ and $y$. So distance is typical weighted Euclid distance.

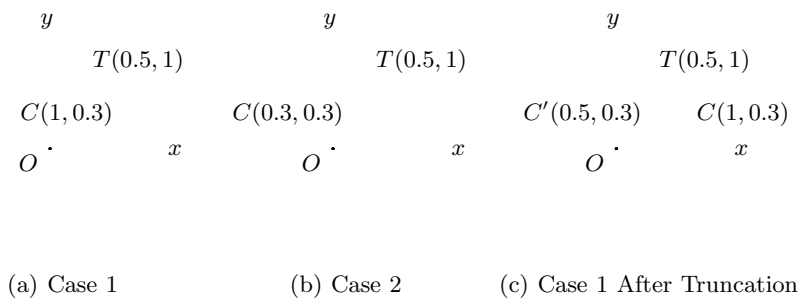
\begin{figure}
\centering
\subfigure[Case 1]
{
\begin{tikzpicture}\label{fig:subfig:truncated distance case1}
\draw[->] (-1.5,0) -- (1.5,0);
\draw[->] (0,-1.5) -- (0,1.5);
\node[right] at (1.5,0) {$x$};
\node[above] at (0,1.5) {$y$};
\node[below=5pt,left] at (0,0) {$O$};
\fill (0.5,1) circle (1.2pt) (1,0.3) circle (1.2pt);
\node[above=5pt,right] at (0.5,1) {$T(0.5,1)$};
\node[above=5pt,left] at (1,0.3) {$C(1,0.3)$};
\end{tikzpicture}
}
\subfigure[Case 2]
{
\begin{tikzpicture}\label{fig:subfig:truncated distance case2}
\draw[->] (-1.5,0) -- (1.5,0);
\draw[->] (0,-1.5) -- (0,1.5);
\node[right] at (1.5,0) {$x$};
\node[above] at (0,1.5) {$y$};
\node[below=5pt,left] at (0,0) {$O$};
\fill (0.5,1) circle (1.2pt) (0.3,0.3) circle (1.2pt);
\node[above=5pt,right] at (0.5,1) {$T(0.5,1)$};
\node[above=5pt,left] at (0.3,0.3){$C(0.3,0.3)$};
\end{tikzpicture}
}
\subfigure [Case 1 After Truncation]
{
\begin{tikzpicture}\label{fig:subfig:truncated distance case3}
\draw[->] (-1.5,0) -- (1.5,0);
\draw[->] (0,-1.5) -- (0,1.5);
\draw[densely dotted,->](1,0.3) -- (0.5,0.3);
\node[right] at (1.5,0) {$x$};
\node[above] at (0,1.5) {$y$};
\node[below=5pt,left] at (0,0) {$O$};
\fill (0.5,1) circle (1.2pt) (1,0.3) circle (1.2pt) (0.5,0.3) circle (1.2pt);
\node[above=5pt,right] at (0.5,1) {$T(0.5,1)$};
\node[above=5pt,right] at (1,0.3) {$C(1,0.3)$};
\node[above=5pt,left] at (0.5,0.3){$C'(0.5,0.3)$};
\end{tikzpicture}
}
\caption{Example of Truncated Distance}
\label{fig:truncated distance example}
\end{figure}

For a clear expression, we define a 0-1 truncating vector $\overrightarrow{TV}(tv_1,tv_2,...,tv_d)$ to describe the truncated distance. Denote $\mathrm{diff_i}$ as the difference on dimension \textit{i} and $\overrightarrow{TV}(i)$ as the \textit{i}th component of $\overrightarrow{TV}$. Truncated difference $\mathrm{\widetilde{diff_i}}$ on ith dimension is:
\begin{equation}\label{eq:truncated difference i}
\mathrm{\widetilde{diff_i}} = \mathrm{diff_i} \times \overrightarrow{TV}(i) = (t_i-c_i)\times \overrightarrow{TV}(i)
\end{equation}
Notice that dimensions where $t_i-c_i<0$  will be referred as strong dimensions of $C$, remaining ones will be referred as weak dimensions accordingly.

According to \eqref{eq:truncated difference i}, the truncated weighted Euclid distance $\widetilde{\mathrm{dis}}$ is:
\begin{equation}\label{eq:twEuclidDis}
\widetilde{\mathrm{dis}} = \sqrt{\sum_{i=1}^d{(w_i\mathrm{\widetilde{diff_i}})^2}}
\end{equation}
Distance measurements in this paper are all truncated distance. $\overrightarrow{TV}$ is referred as truncating vector henceforth. For better presentation, we denote $\mathrm{\widetilde{oDis}}$ as the truncated distance between two objects which is calculated using \eqref{eq:twEuclidDis}.
\subsubsection{Exchange Procedure}\label{sec:exchange procedure}
Define the exchange procedure as swapping $R(r_1,r_2,...,r_d)$ in $C$ with $P(p_1,p_2,..,p_d)$ in $\mathbb{O}$, thus new $\mathrm{diff_i}'$ is: 
\begin{equation}\label{eq:new diff after exchange}
\mathrm{diff_i}' = t_i-(c_i-r_i+\frac{\lambda_r}{\lambda_p}p_i)
\end{equation}
where $\lambda_r,\small \lambda_p \in \Lambda$ are exchange parameters defined in section \ref{sec:problem formulation}.

Accordingly, truncated difference $\widetilde{\mathrm{diff_i}}'$ after exchanging procedure can be calculated by \eqref{eq:truncated difference i} with a new 0-1 truncating vector $\overrightarrow{TV_1}$ based on situation on each dimension. Hence, new $\widetilde{\mathrm{dis}}'$ after exchanging objects can be calculated with $\widetilde{\mathrm{diff_i}}'$ by applying \eqref{eq:twEuclidDis}.

\subsection{RTC* Method}\label{section:RTC* method model}
Before we propose the RTC* method, we define a virtual object as follows:
\begin{definition}(Virtual Object)
Define a virtual object $V(v_1,v_2,..v_d)$ which could make $C$ has the same value of $T$ on each weak dimension after exchanging with object $R$ in $C$. Thus, value of virtual object on dimension i is:
\begin{equation}\label{eq:vp estimate}
v_i =\frac{\mathrm{diff_i}+r_i}{\lambda_r}\times \overrightarrow{TV_2}(i)
\end{equation}
where $\overrightarrow{TV_2}$ is new truncating vector for virtual object and $r_i$ is the value of swap-out object $R$ on dimension i.
\end{definition}
\begin{corollary}\label{cor:distance cor}
Assume $\forall w_i \in W(w_1,w_2,...,w_d),\small w_i>0$, denote the truncated distance between objects as $\widetilde{\mathrm{oDis}}$. The nearest neighbours of virtual objects measured by $\widetilde{\mathrm{oDis}}$ is the top-ranked ones who can make $C$ become closer to $T$.
\end{corollary}
\begin{proof}
Suppose we can find a nearest neighbour $P$ of $V$, $\lambda_pP \in \mathbb{O}$, truncated difference between $V$ and $P$ is represented using $\widetilde{\Delta}$ where $\widetilde{\Delta}(i)$ is the truncated value on dimension \textit{i}. Thus, $\widetilde{\mathrm{diff}}'$ is:
\begin{equation}\label{eq:truncated difference of vp}
\widetilde{\mathrm{diff_i}}'= (\frac{\mathrm{diff_i}+r_i}{\lambda_r}-p_i)\times \overrightarrow{TV_3}(i)
\end{equation}
where $\overrightarrow{TV_3}$ is truncating vector and $\overrightarrow{TV_3}(i)=0$ $\mathit{iff}$ $ \widetilde{\mathrm{\Delta}(i)}<0$, so $\widetilde{\mathrm{dis}}'$ can be represented as:
\begin{equation}\label{eq:prove same}
\widetilde{\mathrm{dis}}'= \sqrt{\sum_{i=1}^{d}{(w_i\widetilde{\mathrm{diff_i}}')^2}} 
\end{equation}
Notice that $\forall w_i \in W(w_1,w_2,...,w_d),\small w_i>0$, \eqref{eq:prove same} also can be represented as:
\begin{equation}\label{eq:prove same vi pi}
\begin{split}
\widetilde{\mathrm{dis}}' 
&= \sqrt{\sum_{i=1}^{d}{(w_i((v_i-p_i)\overrightarrow{TV_3}(i))^2}} \\
&= \sqrt{\sum_{i=1}^{d}{(w_i(\widetilde{\Delta}(i) \times \overrightarrow{TV_3}(i))^2}} \\
&= \widetilde{\mathrm{oDis}}
\end{split}
\end{equation}
\qed
\end{proof}

So our problem of ranking objects from perspective of team context can be mapped into object space. Which is, by considering  nearest neighbours of virtual objects under current team context, we can obtain top-ranked tuples.

We can index the truncated distance between $\forall O_i \in \mathbb{O}$ and the virtual object for convenience of searching:
\begin{algorithm}
\caption{RTC* Method}\label{algo:RTC* method}
\KwIn{current context $C$,Target Context $T$,$\mathbb{O}$}
\KwOut{$<R_i,P_i>$}
\BlankLine
\lForEach{$R_i \in C$}{Calculate $V_i$ using \eqref{eq:vp estimate}}\;
\lForEach{$P_i \in \mathbb{O}$}{Index $\widetilde{\mathrm{oDis}}$ between $P_i$ and $V_i$}\;
find $<R_i,P_i>$ with $\mathrm{Min}(\widetilde{\mathrm{oDis}})$\;
return $<R_i,P_i>$\;
\end{algorithm}

As presented in Algorithm \ref{algo:RTC* method}, we first calculate the virtual object based on current team context and index the $\widetilde{\mathrm{oDis}}$ between $O \in \mathbb{O}$ and virtual object in iDistance presented in \cite{DBLP:journals/infsof/YuCWS07} for processing the query.

It is easy to make the generalization that the query time is only related to the cardinality of our current context $C$, so RTC* will show high performance and good scalability on very large datasets.

\section{Experiments}\label{sec:Experiment}
\subsection{Experiment Setup}
All the experiments were performed on machine with Intel Core(TM) i3 CPU and 4 GB RAM hosted on 32 bit Windows 7.
\paragraph{Datasets}
We perform our experiments on both real and synthetic data. Real dataset is obtained from \cite{website:basketball-reference} which consists total statistic data of NBA regular season from 2011 to 2012.  Real dataset contains 400 players with 24 attributes in total and and 30 teams described by 20 dimensions. Size of player dataset is 39.5KB and team dataset is 20KB.

Attributes which can discriminate between season-long successful and unsuccessful basketball teams according to researches on basketball in \cite{doi:10.1080/17461390802261470,oliver2004basketball} are FG, 3P, 3PA, BLK, FT, STL, FTA, PTS, AST, DRB and TRB. We use this attribute set for our experiments as well.

Synthetic dataset are generated based on the features of real dataset with total $1.07\times10^6$ records and 69MB size. Feature of partial dimensions is illustrated in Fig.\ref{fig:Distributed Fitting Curve}. We make hypothesis $H_0$ that values of dimensions listed in \ref{Table:distribted parameters} has negative binomial distribution and do distribution fitting accordingly.

\begin{figure} 
  \centering
  \subfigure[Player PTS]{
    \label{fig:subfig:b} 
    \includegraphics[width=2in,height=1.5in]{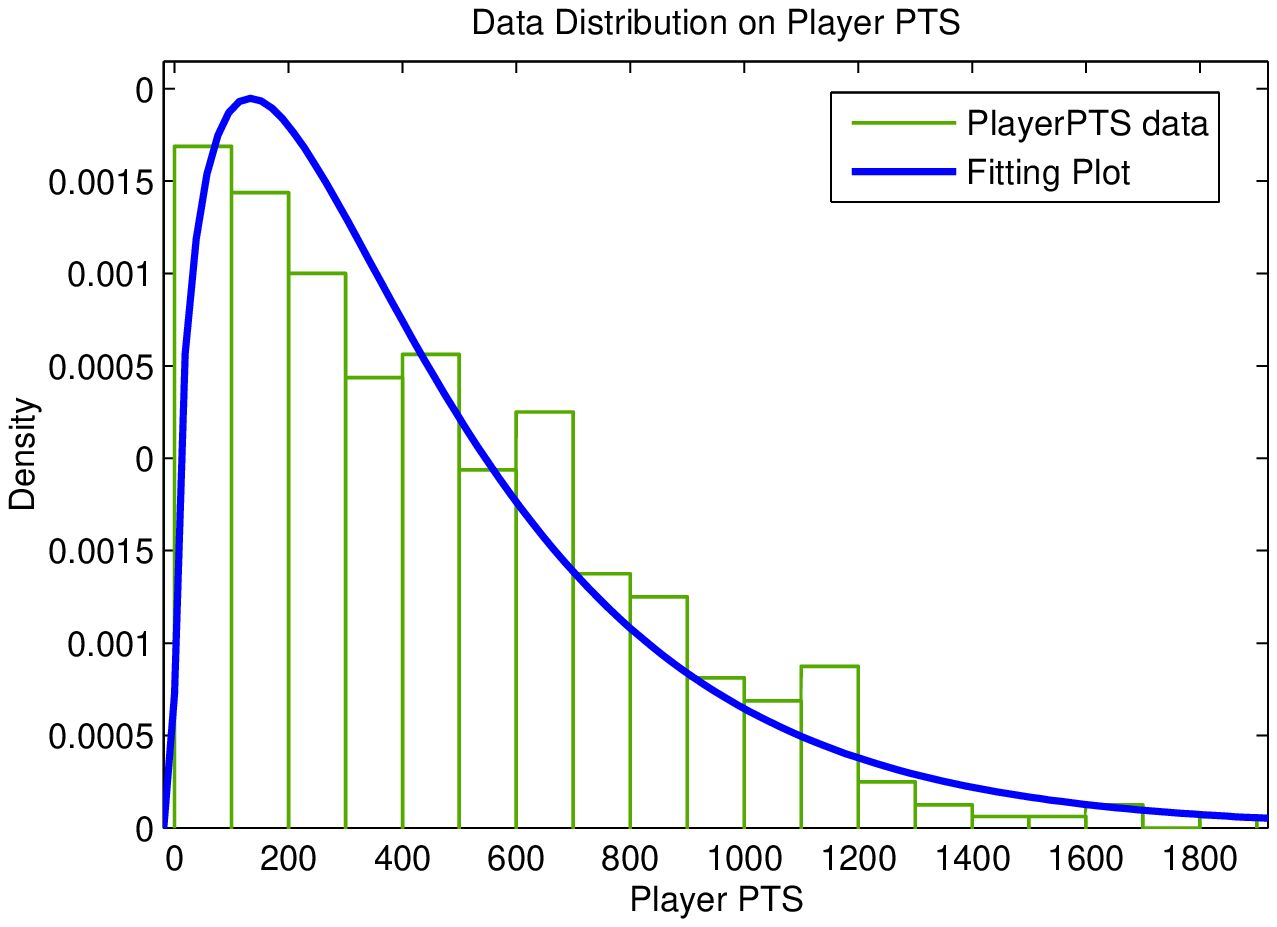}}
  \subfigure[Player STL]{
    \label{fig:subfig:c} 
    \includegraphics[width=2in,height=1.5in]{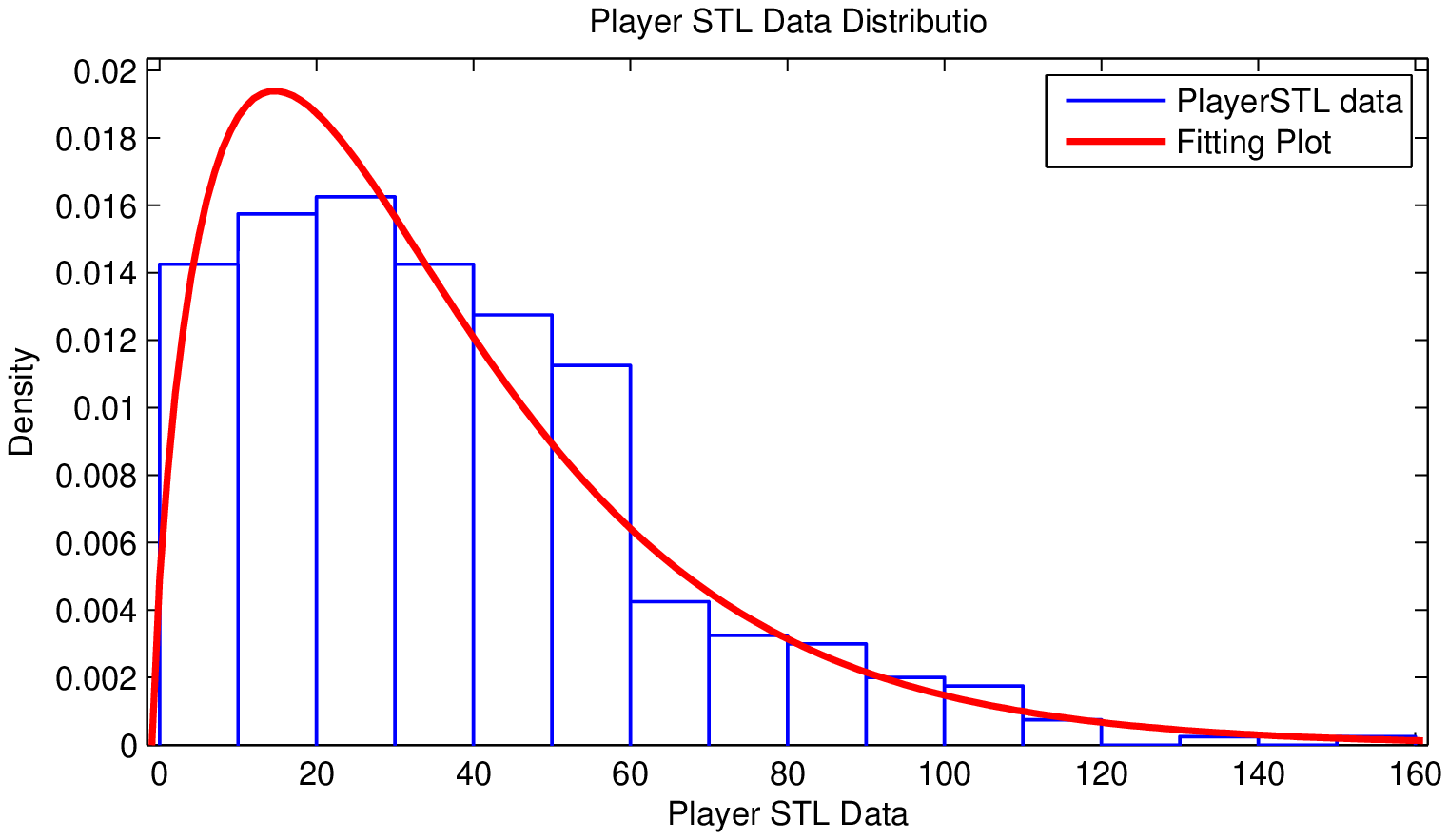}}
  \caption{Partial Review of Data Distribution}\label{fig:Distributed Fitting Curve}
\end{figure}

We test the hypothesis using Chi-square goodness-of-fit with parameters estimated in Table \ref{Table:distribted parameters}. $H_0$ is accepted at 95\% significance level. Synthetic data are generated based on the fitted distribution.

 \begin{table}
 \centering
\caption{Estimated Parameters of Data Distribution}
\begin{tabular}{ccccccccc}\toprule
Dimension& r &p &Dimension &r & p&Dimension &r & p\\ \midrule
FG &1.44& 0.008 &TRB  & 1.62 & 0.008&BLK&0.91& 0.004\\
DRB& 1.67 & 0.01 &FT  & 1.07 & 0.013&STL&1.70 &0.045 \\
FTA& 1.16 & 0.01 &PTS & 1.40 & 0.003&AST&0.93 & 0.0092 \\
\bottomrule\end{tabular}
\label{Table:distribted parameters}
\end{table}

Because each dimension contributes differently in teams' final rankings, we adopt Kendall's $\tau$, which is a method measuring the associations between attributes \cite{kendall1938new}, to calculate weight of each dimension. Results are listed in Table \ref{Table:weight}.
\begin{table}
\centering
\caption{Weight for Each Dimension}
\begin{tabular}{clclclcl}\toprule
Dimension&Weight&Dimension&Weight&Dimension&Weight&Dimension&Weight\\ \midrule
DRB&0.35&FG&0.2695&3P&0.30&AST&0.24\\
3PA&0.20&FT&0.2576&TRB&0.1884&STL&0.38\\
FTA&0.27&PTS&0.4060&BLK&0.24\\
\bottomrule\end{tabular}
\label{Table:weight}
\end{table}
\subsection{Experiments Implementations}\label{sec:Experiments Implementations}
First, we select the target team on real dataset before exchanging players. Results are shown in Table \ref{Table:Target Context-r} with the initial truncated distance between $C$ and $T$.
\begin{table}
\centering
\caption{Target Context of Each Team}
\begin{tabular}{ccccccccc} \toprule
$C$ & $T$ & Distance & $C$ & $T$ & Distance & $C$ & $T$ & Distance \\ \midrule
DEN &BOS &6.1096   &PHI& BOS &29.5286&UTA &MEM &13.6334 \\
ORL &BOS & 51.2774 &HOU& ATL &31.2126&DAL &ATL &28.4675 \\
NYK &MEM & 23.6467 &PHO& BOS &18.3673&MIL& MEM &19.3955\\
POR& LAC &31.0965\\
\bottomrule\end{tabular}
\label{Table:Target Context-r}
\end{table}
\subsubsection{A Brute Force Method}\label{sec:brute force}
Brute force method is performed for each "mid-class"(teams ranked 11st$\scriptsize{\sim}$20th based on game winning in season 2011$\scriptsize{\sim}$2012) team on real dataset as a baseline method. 
\begin{table}
\centering
\caption{Selection For HOU}
\begin{tabular}{llc}\toprule
Roster & Candidate & New Distance\\ \midrule
Luis Scola& Josh Smith&0\\
Patrick Patterson& LeBron James  &0\\
\bottomrule\end{tabular}
\label{Table:Object Selection Using BF}
\end{table}
Take results of HOU listed in Table \ref{Table:Object Selection Using BF} as an example. There are two pairs of players could be found for making this team approach its target measured by truncated distance. Either pair can be selected to make chance for this team to enter into play-offs.

\begin{figure}
  \centering
  \includegraphics[scale=0.5]{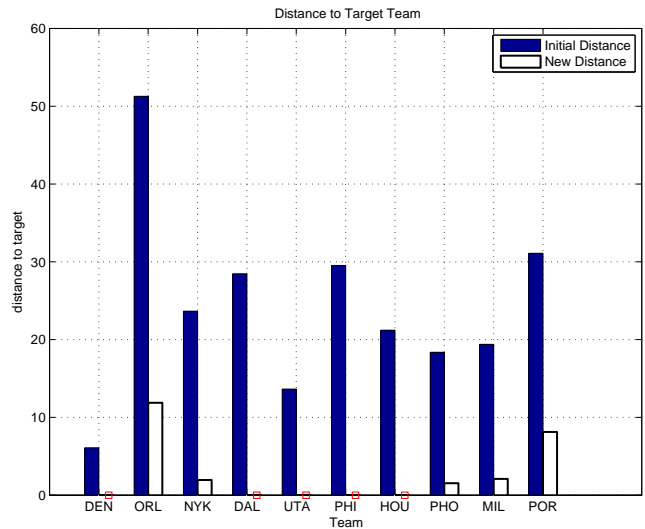}
  \caption{Comparison of Team's Distance to Target}
  \label{fig:new distance}
\end{figure}

Fig.\ref{fig:new distance} depicts between each "mid-class" team and its target before and after exchanging players. We can observe that many teams will have same value as its targets on their weak dimensions, which are illustrated in square.

\subsubsection{RTC* Method}
We also test RTC* method on the real dataset.  In estimating virtual players using \eqref{eq:vp estimate}, minutes played serves as the exchange parameter $\lambda_r$. 

Also consider HOU(Houston Rockets) as an example. Values on dimension FG,3P,3PA,FT and FTA  of virtual player and corresponding swap-in player listed in Table \ref{Table:Object Selection Using RTC*} for explanation:

\begin{table}
\caption{Overview of Virtual Player and Candidate}
\label{Table:Object Selection Using RTC*}
\centering
\begin{tabular}{l|c|c|c|c|c|c}
\toprule
\multirow{2}{*}{Name}  &
\multicolumn{5}{|c|}{Attributes} &
\multirow{2}{*}{ $\widetilde{\mathrm{oDis}}$} \\ 
\cline{2-6}
 {}& FG & 3P & 3PA & FT & FTA  &{} \\
\hline
 Josh Smith&0.22&0.01&0.05&0.09&0.14&
 \multirow{2}{*}{0} \\
 Virtual Player of Josh&0.18&0.01&0.00&0.08&0.14\\ \hline
 LeBron James&0.27&0.02&0.06&0.17&0.22&
 \multirow{2}{*}{0} \\
 Virtual Player of LeBron&0.11&0.01&0.00&0.03&0.10\\
\bottomrule
\end{tabular} 
\end{table}

As listed in Table \ref{Table:Object Selection Using RTC*}, both selected players from player space are better than corresponding virtual players on those dimensions. According to definition of truncated distance, the $\widetilde{\mathrm{oDis}}$ calculated between Josh Smith and its corresponding virtual player, or LeBron James and its virtual player listed, are 0. Therefore, those two pairs are selected as the result which is same as the ones selected using brute force listed in Table \ref{Table:Object Selection Using BF}. 

Notice that both results are the nearest neighbours of corresponding virtual player measured by $\widetilde{\mathrm{oDis}}$, which further proves the rationale behind Corollary \ref{cor:distance cor}.

\subsection{Result Analysis}In this section, we mainly focus on analysing the results of brute force method and RTC* on efficiency and scalability.

We fixed our block size as 100 records per block for real dataset and 10 per block for synthetic data.
\begin{figure}
  \centering
  \subfigure[I/O Test on real data]{
    \label{fig:subfig:BFNNIO} 
    \includegraphics[width=2.2in,height=1.8in]{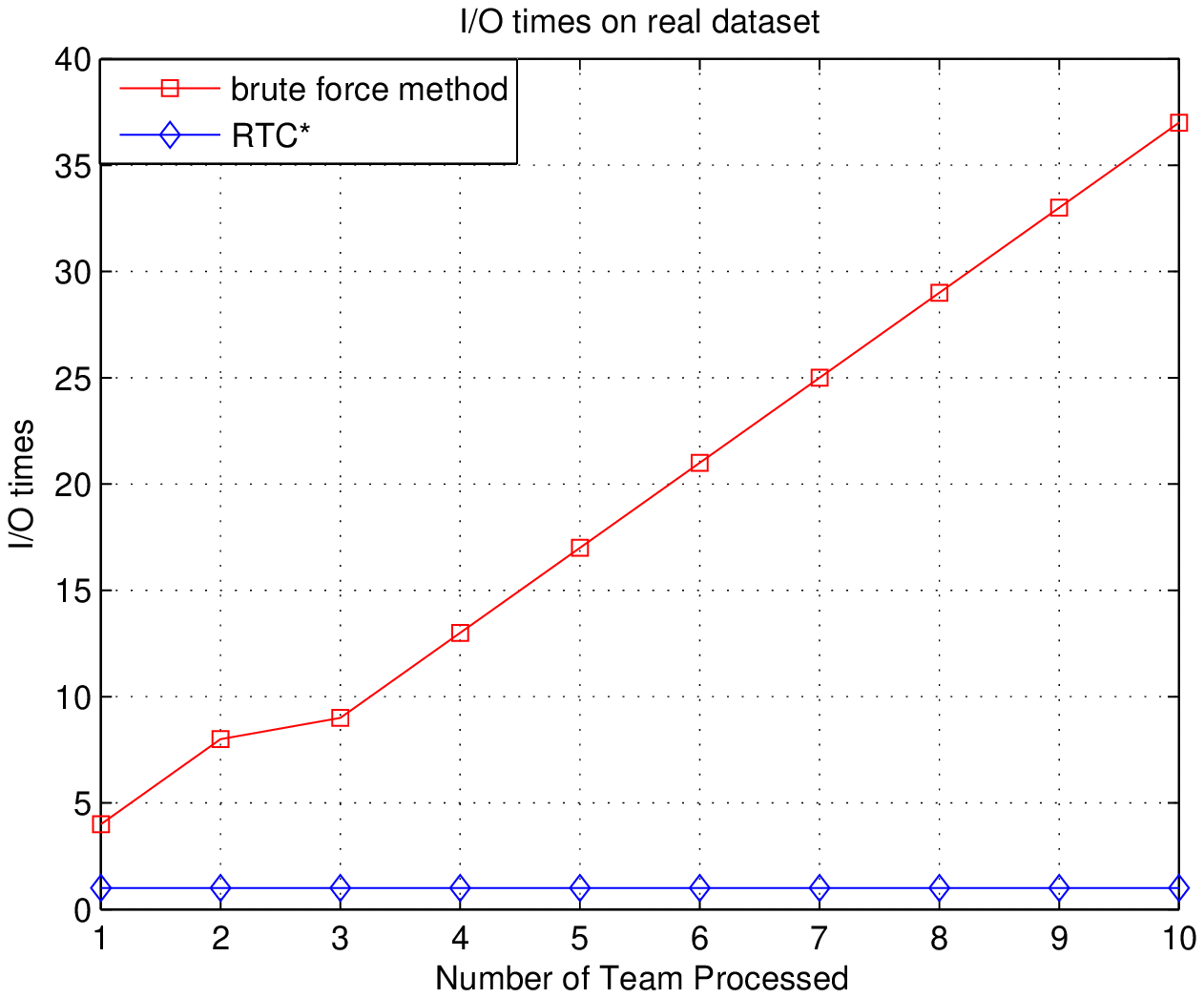}}
  \subfigure[I/O Test on synthetic data]{
    \label{fig:subfig:BFNNIO_Syn} 
    \includegraphics[width=2.2in,height=1.8in]{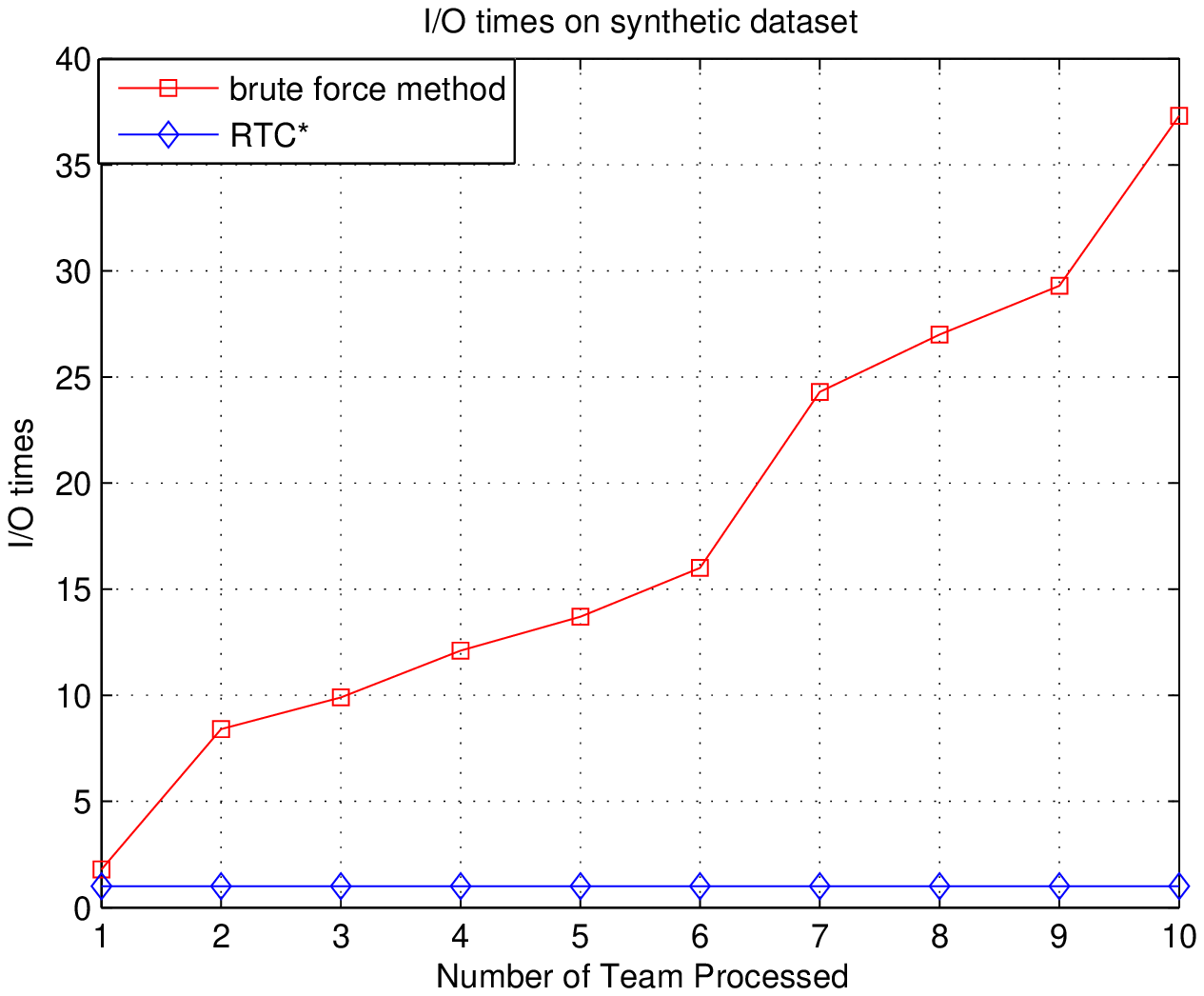}}
  \caption{I/O Performance Of Different Method }
  \label{fig:I/O Testing}
\end{figure}

%
Depicted in Fig.\ref{fig:I/O Testing}, regardless of the data size and block size, I/O will be performed only once using RTC* as long as we had set up virtual player index, while brute force method varies depending on TC, which shows less robustness. 

Like I/O testing, we tested time cost of selection method both on real and synthetic data, illustrated in Fig.\ref{fig:TimeTesting}.
\begin{figure}
  \centering
  \subfigure[Time Cost on Real Data]{
    \label{fig:subfig:d} 
    \includegraphics[width=2.2in,height=1.8in]{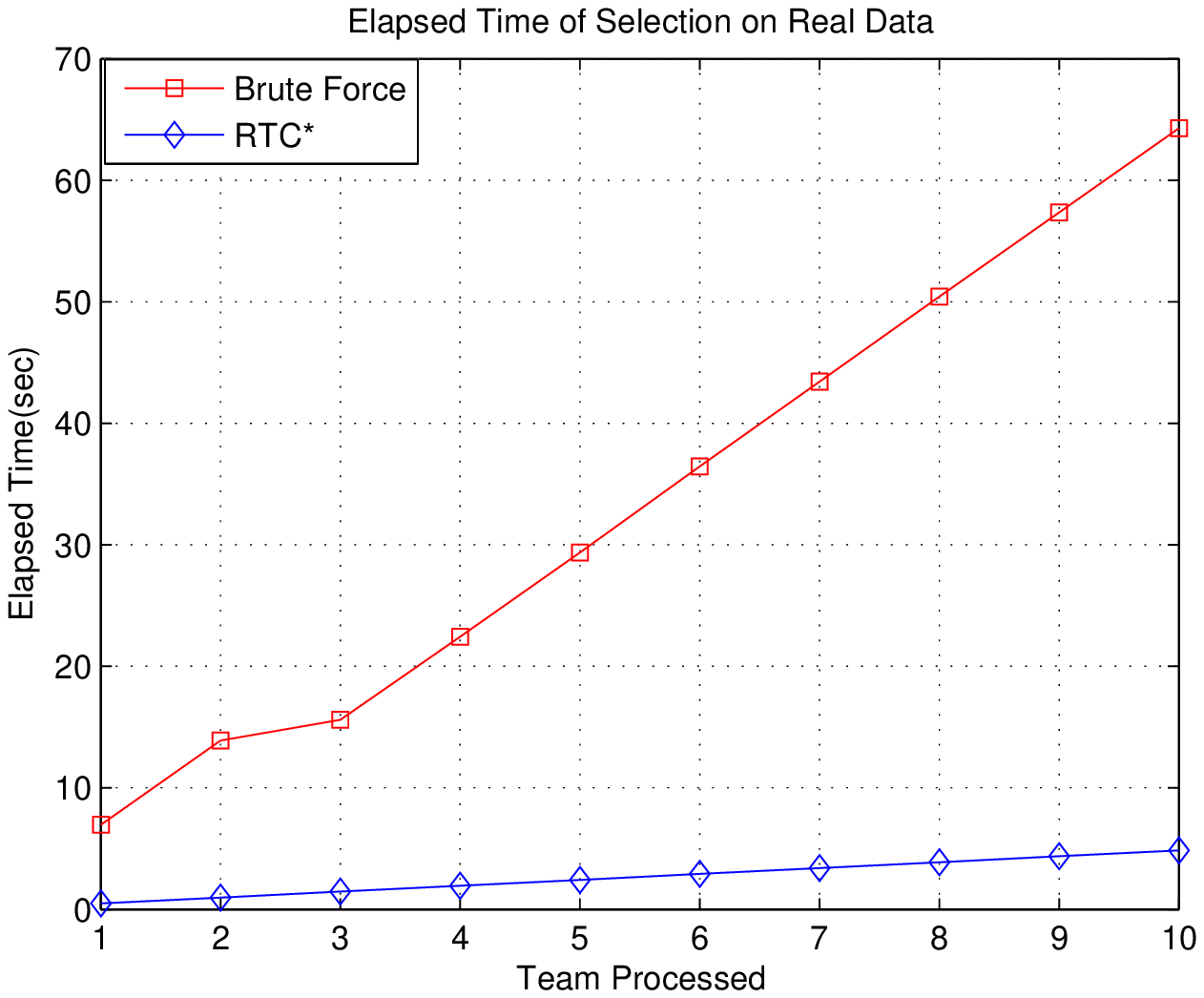}}
  \subfigure[Time Cost on Synthetic]{
    \label{fig:subfig:e} 
    \includegraphics[width=2.2in,height=1.8in]{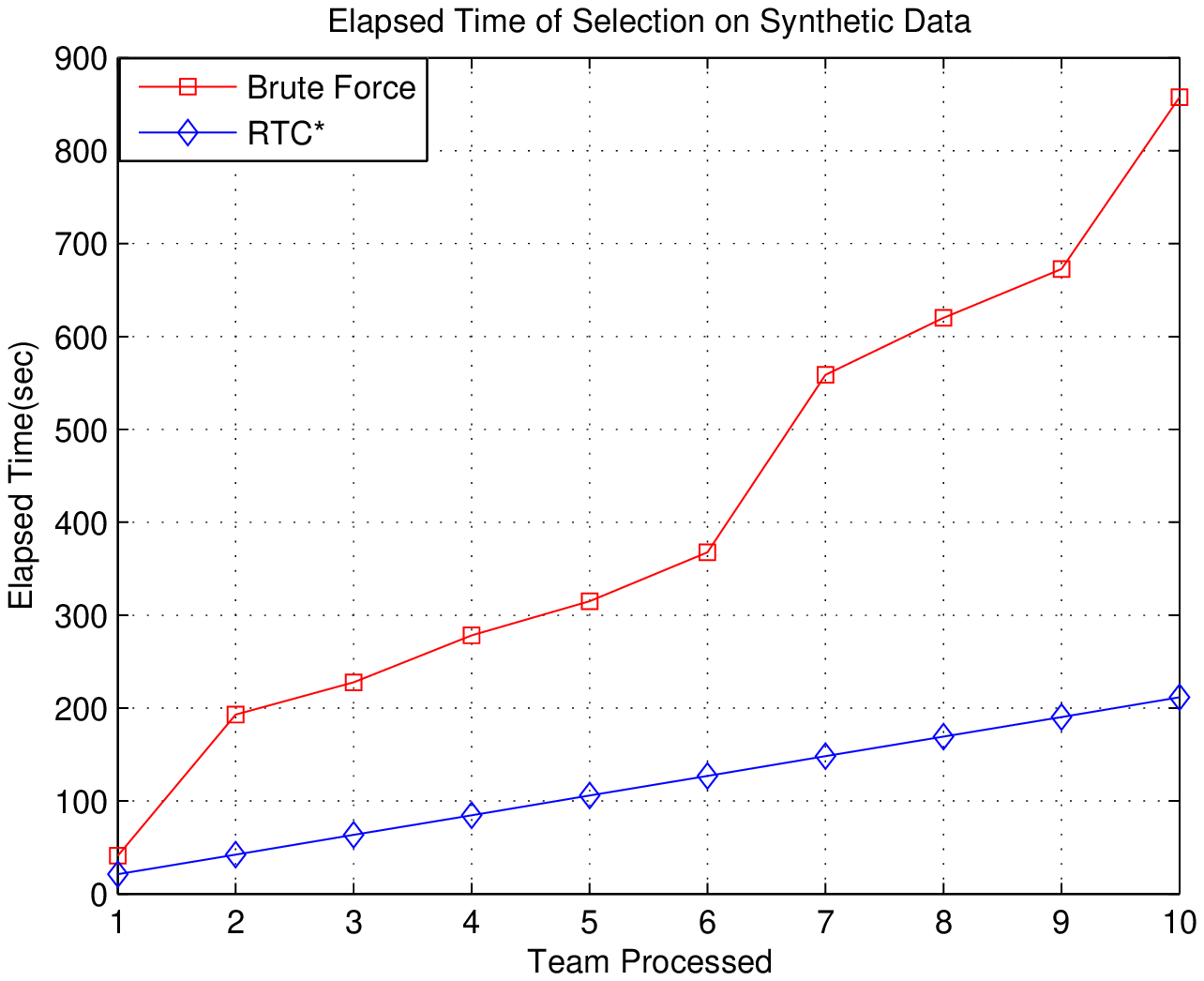}}
  \caption{Time Cost Of Different Method }
  \label{fig:TimeTesting}
\end{figure}

It is easy to generalize that the time cost of RTC* grows slowly with a constant rate while brute force increase very fast. 

Differs from brute force method which highly depends on the value of team context, RTC* has good robustness and better performance regardless of the data value.

\section{Conclusion and Future Work} \label{sec:ConclusionLB}
In this paper, we introduce the problem of object selection under team context. This problem is quite practical in many scenarios of selecting objects to improve the team or organization's competence. We propose the brute force algorithm RTC for the problem. Furthermore, we propose an I/O efficient algorithm RTC* based on NN-indexing with a proof that its output is equivalent to RTC. Extensive experiments are conducted on both synthetic and real datasets to demonstrate the effectiveness and efficiency of our algorithms.

We would like to extend our work from two directions in our future work. First, due to the fact that the probabilistic database tuples are not uncommon, we plan to do probabilistic object selection.  Second, query the database objects based on the teams' temporal contexts.
\bibliographystyle{abbrv}
\bibliography{draft_1401}
\end{document}